%% file: DMSG_12.tex
\documentclass[a4paper]{jpconf}
\usepackage{graphicx}
\usepackage{amsfonts}

\newcommand{\calZ}{{\mathcal Z}}
\newcommand{\calM}{{\mathcal M}}

\newcommand{\calC}{{\mathcal C}}

\newcommand{\bsh}{{\backslash}}

\newcommand{\N}{{\mathbb N}}
\renewcommand{\P}{\N_{>0}}

\newcommand{\Q}{{\mathbb Q}}
\newcommand{\R}{{\mathbb R}}
\newcommand{\C}{{\mathbb C}}

\newcommand{\Frac}[2]{\displaystyle \frac{#1}{#2}}
\newcommand{\Sum}[2]{\displaystyle{\sum_{#1}^{#2}}}
\newcommand{\Prod}[2]{\displaystyle{\prod_{#1}^{#2}}}

\newcommand{\eps}{{\varepsilon}}

\def\Lyn{{\mathcal Lyn}}

\def\CX{\C \langle X \rangle}

\def\ad{\mathrm{ad}}

\def\AXX{\serie{A}{X}}
 \def\shuffle{\mathop{_{^{\sqcup\!\sqcup}}}}

\gdef\stuffle{\;%
  \setlength{\unitlength}{0.0125cm}%
  \begin{picture}(20,10)(220,580)
  \thinlines
  \put(220,592){\line( 0,-1){ 10}}
  \put(220,582){\line( 1, 0){ 20}}
  \put(240,582){\line( 0, 1){ 10}}
  \put(230,592){\line( 0,-1){ 10}}
  \put(225,587){\line( 1, 0){ 10}}
  \end{picture}\;
}

\newtheorem{corollary}{Corollary}
\newtheorem{proposition}{Proposition}
\newtheorem{theorem}{Theorem}
\newtheorem{lemma}{Lemma}
\newtheorem{definition}{Definition}
\newtheorem{conjecture}{Conjecture}
\newtheorem{example}{Example}

\newtheorem{state}[corollary]{Statement}

\newcommand{\Li}{\mathrm{Li}}

\def\L{\mathrm{L}}
\def\H{\mathrm{H}}

\def\P{\mathrm{P}}

\def\L{\mathbb{L}}
\def\V{\mathbb{V}}

\def\bv{\mid}

\def\pointir{\unskip . --- \ignorespaces}
\def\diag{\mathbf{diag}}
\def\ldiag{\mathbf{ldiag}}

\def\LDIAG{\mathbf{LDIAG}}
\def\DIAG{\mathbf{DIAG}}

\def\MQS{\mathbf{MQSym}}
\def\ncp#1#2{#1\langle #2\rangle}
\def\al{\alpha}
\def\be{\beta}
\def\ncp#1#2{#1\langle #2\rangle}
\def\ra{\rightarrow}

\newcommand{\poly}[2]{#1 \langle #2 \rangle}

\def\QX{\poly{\Q}{X}}

\def\CX{\poly{\C}{X}}

\def\QY{\poly{\Q}{Y}}

\newcommand{\serie}[2]{#1 \langle \! \langle #2 \rangle \! \rangle}

\def\CXX{\serie{\C}{X}}
\def\CYY{\serie{\C}{Y}}

\def\Lie{{\cal L}ie}

\def\LQX{\Lie_{\Q} \langle X \rangle}
\def\LQY{\Lie_{\Q} \langle Y \rangle}

\def\LAXX{\Lie_{A} \langle\!\langle X \rangle \!\rangle}
\def\LCXX{\Lie_{\C} \langle\!\langle X \rangle \!\rangle}

\def\PKZ{\Phi_{KZ}}

\newenvironment{proof}{\par \noindent {\sc Proof}\ --\ }{$\square$ \par}
\def\CX{\C \langle X \rangle}
\def\CXX{\serie{\C}{X}}

\newcommand{\calT}{{\cal T}}

\def\Sum{\displaystyle\sum}
\def\Prod{\displaystyle\prod}
\def\Frac{\displaystyle\frac}
\def\path{\rightsquigarrow}
\def\reg{\mathop\mathrm{reg}\nolimits}

\gdef\minishuffle{{\scriptstyle \shuffle}}
\gdef\ministuffle{{\scriptstyle \stuffle}}

\begin{document}
\title{An interface between physics and number theory}
\author{G\'erard H.E. Duchamp$^1$,  Hoang Ngoc Minh$^{1,2}$, Allan I. Solomon$^{3,4}$, Silvia Goodenough$^1$}
\address{
$^1$ Universit\'e Paris 13 -- LIPN - UMR 7030 CNRS,
F-93430 Villetaneuse, France.}
\address{
$^2$ Universit\'e Lille 2,
1 Place D\'eliot, 59024 Lille C\'edex, France.}
\address{
$^3$ Physics and Astronomy Department,The Open University,
Milton Keynes MK7~6AA, UK}
\address{
$^4$
Universit\'e Pierre et Marie Curie, LPTMC,  UMR 7600 CNRS,
4 pl. Jussieu,  75252 Paris 5, France.}
\begin{abstract}
We extend the Hopf algebra description of a simple quantum system given previously, to a more elaborate Hopf algebra, which is rich enough to encompass that related to a description of perturbative quantum field theory (pQFT). This provides a {\em mathematical} route  from an algebraic description of non-relativistic, non-field theoretic quantum statistical mechanics to one of relativistic quantum field theory.

Such a description necessarily involves treating the algebra of polyzeta functions, extensions of the Riemann Zeta function, since these occur naturally in pQFT.  This provides a link between physics, algebra and number theory. As a by-product of this approach, we are led to indicate  {\it inter alia} a basis for concluding that  the Euler gamma constant $\gamma$ may be rational.
\end{abstract}
\section{Introduction}
In an introductory paper delivered at this Conference\footnote{Talk delivered at the 28th International Colloquium on Group Theoretical Methods in Physics, Northumbria University, Newcastle, July  2010.}, {\it ``From Quantum Mechanics to Quantum Field Theory:
The Hopf route''}, Allan I. Solomon \cite{AISG28,BDSHP} {\it et al.} start their
exposition with the Bell numbers $B(n)$ which count the number of {\it set partitions} within a given
set of $n$ elements. It is shown there that these very elementary combinatorial ideas are in a sense generic within quantum statistical mechanics. Further, it turns out that this approach also leads to a Hopf algebraic description of the physical system. Since recent work on relativistic field theory (perturbative quantum field theory or pQFT) also leads to a Hopf algebra description, the question naturally arises as to whether the algebraic structure,  implicit  in essentially the non-relativistic commutation relations, may be extended to the more complicated, relativistic system. Instead of extending the physical picture, we have chosen to follow a mathematical route. Doing so, it turns out that the basic structure, as exemplified in
{\it ``The Hopf route''} for the special case of Bell number diagrams,  provides the building blocks of a non-commutative parametric Hopf
algebra which specializes on the one hand to a Hopf algebra of noncommutative
symmetric functions, and on the other hand to a proper Hopf sub-algebra of that inherent in perturbative quantum field theory (pQFT).  It is well known that analysis of pQFT reveals that polyzeta functions, extensions of the Riemann zeta function of number theory, play an important role.  It is not surprising therefore that the algebraic analysis embarked upon here should reveal properties of the polyzetas.  And indeed it is the case that  our  algebra   admits,
as a homomorphic image, the algebra of polyzeta functions of number theory.  This foray into   number theory leads us to consider Euler's gamma constant $\gamma$, and conclude, on the basis of Lemmas given in the text,  that $\gamma$   may in fact be a rational number.

\section{One-parameter groups and the exponential Hadamard product}

As often in physics, one has to consider one-parameter groups with infinitesimal
generator $F$,
\begin{equation}
	G(t):=e^{tF}\ .
\end{equation}
When we have two infinitesimal generators which commute $[F_1,F_2]=0$ (for example, with functions constructed from   two distinct boson modes), one can obtain the group $G_{12}(t):=e^{t F_1 F_2}$ with  operation
\begin{equation}
	G_{12}(t):=G_1(t\frac{\partial}{\partial x})[G_2(x)]\big|_{x=0}\ .
\end{equation}
where $G_i(t):=e^{tF_i},\ i=1,2$. Up to  constant terms, this operation is
a {\em Hadamard product} (i.e. the componentwise product, which we denote below by $\mathcal{H}$)
 but performed on the exponential generating
series \cite{GOF18}. Specifically with
\begin{eqnarray}\label{ExpHad}
F(z)=\sum_{n\geq 0} a_n\frac{z^n}{n!},\;\;\; G(z)=\sum_{n\geq 0}
b_n\frac{z^n}{n!},\;\;\;\;
\mathcal{H}(F,G):=\sum_{n\geq 0}
a_nb_n\frac{z^n}{n!}\ .
\end{eqnarray}
This formula is known as the {\it product formula} \cite{BBM}.
Since one-parameter groups are exponentials, we can try to find a universal formula for
(\ref{ExpHad}). To this end we set
\begin{eqnarray}\label{setting1}
F(z)=\exp\left(\sum_{n=1}^\infty L_n\frac{z^n}{n!}\right),\
G(z)=\exp\left(\sum_{n=1}^\infty V_n\frac{z^n}{n!}\right).
\end{eqnarray}
and remark that $F,G$ can simultaneously  be expressed in terms of Bell\footnote{Named after the mathematician and science fiction author Eric Temple Bell (1883, Scotland, to 1960, US).} polynomials \cite{comtet}
 $B_n$ as
\begin{equation}\label{bellpol1}
\exp\left(\sum_{n=1}^\infty X_n\frac{z^n}{n!}\right)=\sum_{n=0}^\infty \frac{z^n}{n!} B_n(X_1,X_2,\cdots X_n)\ .
\end{equation}
The coefficients of the polynomials $B_n$ are positive integers which can be interpreted as the number of {\it unordered set partitions} which are collections of mutually disjoint non-empty subsets (blocks) of a (finite) set $F$ whose union is $F$. For example, the {\it set partitions} of $F=\{1,2,3,4\}$ in two blocks of size two are
\begin{equation}
\{\{1, 2\}, \{3, 4\}\}\ \{\{1, 3\}, \{2, 4\}\}\ \{\{1, 4\}, \{2, 3\}\}\ .
\end{equation}
The type of a set partition $P=\{A_1,A_2,\cdots A_k\}$ of $F$ is the sequence of numbers
$$
Type(P):=(\#\{l\ s. t.\ |A_l|=k\})_{k\geq 1}\ .
$$
For example, the type of the set partition $P=\{\{2,4,6,7\},\{1,5\},\{3,8\}\}$ is $(0,2,0,1,0,0,0,\cdots)$.
We have the formula
\begin{equation}\label{bellpol2}
B_n(X_1,X_2,\cdots X_n)=\sum_{P\in UP_n} \mathbb{X}^{Type(P)}	
\end{equation}
(where $UP_n$ is the set of unordered partitions of $[1\cdots n]$) which can be proved (following the
 argument of Fa\`a di Bruno) using the exponential formula \cite{PDGP}.
Now from (\ref{setting1}), (\ref{bellpol1}), (\ref{bellpol2}), one concludes that
\begin{eqnarray}\label{expansion1}
\mathcal{H}(F,G)=\sum_{n\geq 0} \frac{z^n}{n!} \sum_{P_1,P_2\in
UP_n} \L^{Type(P_1)}\V^{Type(P_2)}\ .
\end{eqnarray}
This formula can be rewritten as a diagrammatic expansion using the incidence matrix of the
pairs of set partitions $P_1,P_2$ involved in (\ref{expansion1}) with respect to intersection numbers, i.e.
\begin{equation}
	\mathrm{Matrix\ of}\ (P_1,P_2)=\Big( \# B_1\cap B_2\Big)_{B_1\in P_1,\ B_2\in P_2}\ .
\end{equation}
These matrices, with unordered rows and columns, can be represented by a diagram with the following rule. Draw two rows of spots (white spots on the upper row and black on the lower row) label each black spot by a block of $P_1$ and each white spot by a block of $P_2$ and join two spots (one of each kind) by the number of common elements in the corresponding blocks\footnote{Compare this diagrammatic approach  with the white and black spots of the preceding presentation\cite{AISG28}.}. The result is the  diagram of the following figure.

\input{unlabelled11}

The set of such diagrams will be denoted by $\diag$. This is precisely the set of bipartite graphs with multiple edges and no isolated vertex (and no order between the spots) and integer multiplicities. It is straightforward to verify that concatenation is associative with unit (the empty graph) within the set $\diag$. We therefore denote the result of concatenation by  $[d_1|d_2]_D$ (arbitrarily  putting $d_2$ on the right of $d_1$  because the law is commutative). The monoid
$(\diag,[-|-]_D,1_{\diag})$ is a {\em free commutative monoid}. Its generators are the connected (non-empty) diagrams ($\diag_c$). The algebra of $\diag$ is the set of formal (finitely supported) sums
\begin{equation}
\{\sum_{d\in \diag}\alpha(d) d\}_{\alpha\in \C^{(\diag)}}
\end{equation}
The freeness of $\diag$ (as a commutative monoid) implies that this algebra is isomorphic to the algebra of polynomials $\C[\diag_c]$. We shall denote it differently (i. e. $\DIAG$) in the sequel (see \cite{GOF18} for details) as we wish to endow it with the structure of a Hopf algebra.\\
Indeed, the diagrams allow us to rewrite formula (\ref{expansion1}) as

\begin{eqnarray}\label{diag_expansion}
\mathcal{H}(F,G)=\sum_{n\geq 0} \frac{z^n}{n!} \sum_{d\in diag\atop
|d|=n} mult(d)\L^{\al(d)}\V^{\be(d)}
\end{eqnarray}
where $\al(d)$ (resp. $\be(d)$) is the ``white spot type'' (resp.
the ``black spot type'') i.e. the multi-index $(\al_i)_{i\in \N^+}$
(resp. $(\be_i)_{i\in \N^+}$) such that $\al_i$ (resp. $\be_i$) is
the number of white spots (resp. black spots) of degree $i$ ($i$
lines connected to the spot) and $mult(d)$ is the number of pairs of
unordered partitions of $[1\cdots |d|]$ (here
$|d|=|\al(d)|=|\be(d)|$ is the number of lines of $d$) with
associated diagram $d$.

\smallskip
Now one may naturally ask \\
{\it Q1) ``Is there a (graphically) natural multiplicative structure
on $\diag$ such that the arrow
\begin{equation}
    d \mapsto \textrm{mon}(d)=\L^{\al(d)}\V^{\be(d)}
\end{equation}
is a morphism ?''}

\smallskip
The answer is ``yes'' concerning the algebraic structure  of $\LDIAG$ as it is easily seen that

\begin{equation}
	\textrm{mon}([d_1|d_2]_D)=\textrm{mon}(d_1)\textrm{mon}(d_2)\ .
\end{equation}

But the algebra of polynomials is also endowed with the  structure of a Hopf algebra (decomposition property \cite{AISG28,BDSHP}). So it is natural to ask whether the coalgebra structure of the polynomials could be lifted to $\DIAG$ and, further, if this allows enriching the structure of $\DIAG$ to a Hopf algebra. Again, the answer is ``yes'' and the appropriate comultiplication is denoted by $\Delta_{BS}$ (``black spot'' coproduct, see \cite{GOF18} for details). Computing $\Delta_{BS}(d)$ is immediate; one simply divides  the black spots into two subsets, then obtain the tensors and sum all the results. More formally; noting  the set of black spots of the diagram $d$ by $BS(d)$ and denoting by $d[I]$ the sub-diagram whose black spots are in $I$ and the white connected to $I$, one has
\begin{equation}\label{coprod_diag}
\Delta_{BS}(d)=\sum_{I+J=BS(d)} d[I]\otimes d[J]	
\end{equation}
For example, one labels by $\{a,b,c\}$ the three black spots of the diagram $d$ of {\bf Fig 1}. In this case, the coproduct reads
\begin{eqnarray*}
&&\Delta_{BS}(d)=d\otimes	[\ ]+ [\ ]\otimes	d+
d[\{a\}]\otimes d[\{b,c\}]+d[\{b\}]\otimes d[\{a,c\}]+\cr
&&d[\{c\}]\otimes d[\{a,b\}]+d[\{a,b\}]\otimes d[\{c\}]+
d[\{a,c\}]\otimes d[\{b\}]+d[\{b,c\}]\otimes d[\{a\}]\ .
\end{eqnarray*}

\subsection{Labelling the nodes and the noncommutative analogue $\LDIAG$}

If one needs to label the black spots in order to compute the comultiplication,
it is obvious that the precise form of labelling is ultimately irrelevant. It is more important to
 label the nodes endowing their sets with a linear order (one for the black spots
among themselves and one for the white spots). This leads to $\LDIAG$, the
(non-deformed) noncommutative analogue of $\DIAG$ \cite{GOF4}.\\
The solution (of this non-deformed problem \cite{GOF4}) is simple and
consists in labelling the black (resp. white) spots from left to right and
from ``$1$'' to $p$ (resp. $q$); thus  one obtains the {\it labelled diagrams}.
Again concatenation is associative and endows the set of the labelled diagrams
(here denoted by $\ldiag$) with the structure of a monoid (the unit is the empty diagram
characterized by $p=q=0$).\\
Again, and with the same formula (\ref{coprod_diag})
(but not applied to the same objects), the algebra of this monoid, $\LDIAG$ and the
counit, which consists in taking the coefficient of the empty diagram, is a Hopf algebra
\footnote{At this stage, it is only a graded bialgebra but, as it is graded and of  finite
dimension, a general theorem states that  an antipode exists; (see \cite{GOF18} formula (19)
for an explicit formula of this antimorphism).}.

\subsection{The deformed case.}

The preceding coding is particularly well adapted to the
deformation we wish to construct here. The philosophy of the
deformed product is expressed by the descriptive formula
\begin{equation} [d_1|d_2]_{L(q_c,q_s)}=
    \sum_{CS(d_1,d_2)}
    q_c^{w1(CS)}q_s^{w2(CS)} CS([d_1|d_2]_{L(q_c,q_s)})
\end{equation}
where
\begin{itemize}
    \item $q_c,q_s\in \C$ or $q_c,q_s$ formal. These and other cases may be unified by considering the set of coefficients as belonging to a ring $K$
    \item the exponent of $w1(CS)$ is the number of crossings of ``what crosses'' times its weight
    \item the exponent of $w2(CS)$ is the product of the weights of ``what is overlapped''
    \item $CS(d_1,d_2)$ are the diagrams obtained from $[d_1|d_2]_L$ by the process of crossing and superposing the black spots of $d_2$ on to those of $d_1$, the order and distinguishability of the black spots of $d_1$ (i.e. $d_2$) being preserved.
\end{itemize}

What is striking is that this law is associative. This fact is by no means trivial and three proofs have been  given (one of them can be found in \cite{GOF18}). It can be shown that this new algebra (associative with unit denoted $\LDIAG(q_c,q_s)  )$ can be endowed with two comultiplications, $\Delta_0$ and $\Delta_1$, such that the result is a Hopf algebra which we denote by  $\LDIAG(q_c,q_s,q_t)$.\\
It can be shown that this deformed Hopf algebra has two interesting specializations
\begin{equation*}
 \LDIAG (0,0,0)\simeq \LDIAG\ \mathrm{ and }\ \LDIAG (1,1,1)\simeq \MQS\ ;
\end{equation*}
$\LDIAG (0,0,0)$ (no crossing and no superposition) is the undeformed case whereas $\LDIAG (1,1,1)$ is $\MQS$, the Hopf algebra of
Noncommutative Matrix Quasisymmetric functions \cite{DHT}.

\subsection{The arrow $\LDIAG(1,1)\ra$ polyzetas.}

We now  explain how a refinement of the ``black spot type'' $\be(d)$ of formula (\ref{diag_expansion}) evaluates the statistics ($\be(d)$) by ``number of nodes with outgoing degree $k$''.  We take advantage of the labelling and count these degrees node by node. Thus, for a labelled diagram $d$, $I(d,k)$ is simply the outgoing degree of the node with label $k$. It turns out, and is straightforward to show, that
$$
I(d):=[I(d,1),I(d,2)\cdots I(d,p)]
$$
($d$ is a diagram with $p$ black spots) is a composition of the integer $|d|$ (number of edges of $d$).\\
Now, we define  two structures which will be ubiquitous in the next section: {\it the shuffle and stuffle algebras}.
Let $X$ be any alphabet (i.e. set of variables). The shuffle algebra is defined on $\ncp{\C}{X}:=\C[X^*]$ by its values on the monomials (here the strings or words) by the following recursion
\begin{eqnarray}
	1_{X^*}\shuffle w&=&w\shuffle 1_{X^*}=w \textrm{ and }\cr
	xu\shuffle yv&=&x(u\shuffle yv)+y(xu\shuffle v)\ .
\end{eqnarray}
Likewise, let $Y={y_i}_{i\geq 1}$ be a countable alphabet indexed by integers $\geq 1$. The stuffle algebra \cite{Ho}
is defined on $\ncp{\C}{Y}:=\C[Y^*]$ by a recursion on the words as
\begin{eqnarray}
	1_{Y^*}\stuffle w&=&w\stuffle 1_{Y^*}=w \textrm{ and }\cr
	y_iu\stuffle y_jv&=&y_i(u\stuffle y_jv)+y_j(y_iu\stuffle v)+
	                    y_{i+j}(u\stuffle v)\ .
\end{eqnarray}
The shuffle (resp. stuffle) operation endows $\ncp{\C}{X}:=\C[X^*]$ (resp. $\ncp{\C}{Y}:=\C[Y^*]$) with the structure of an AAU\footnote{Associative Algebra with Unit: Here they are, moreover, Hopf algebras \cite{Ho} but we will not need this here.}. We define mappings
 \begin{eqnarray}
&&\phi_{X2LDIAG}\ :\ \LDIAG(1,0)\ra \ncp{\C}{X}\cr
&&\phi_{Y2LDIAG}\ :\ \LDIAG(1,1)\ra \ncp{\C}{Y}	
\end{eqnarray}
(here $X=\{x_i\}_{i\geq 1}$) by
$$
\phi_{X2LDIAG}(d)=x_{I(d,1)}x_{I(d,2)}\cdots x_{I(d,p)}
$$
$$
\phi_{Y2LDIAG}(d)=y_{I(d,1)}y_{I(d,2)}\cdots y_{I(d,p)})
$$
and have the following result.

\begin{proposition} The mappings
$$
\phi_{X2LDIAG}\ :\ (\LDIAG(1,0),[\ .|\ .]_{LD(1,0)}\ra (\ncp{\C}{X},\stuffle)
$$	
$$
\phi_{Y2LDIAG}\ :\ (\LDIAG(1,1),[\ .|\ .]_{LD(1,1)}\ra (\ncp{\C}{Y},\stuffle)
$$	
are epimorphisms of AAU's.
\end{proposition}

\bigskip
This provides  our link to  the next section.

\section{Knizhnik-Zamolodchikov differential system}

In 1986, in order to study the linear representation of the braid group $B_n$ \cite{drinfeld}
coming from the monodromy of the Knizhnik-Zamolodchikov differential equations over
$\C_*^n=\{(z_1,\ldots,z_n)\in\C^n|z_i\neq z_j\mbox{ for }i\neq j\}$~:
\begin{eqnarray}\label{KZ}
d F(z)=\Omega_nF(z),
\end{eqnarray}
where
\begin{eqnarray}
\Omega_n=\frac{1}{2{\rm i}\pi}\sum_{1\le i<j\le n}t_{i,j}\frac{d(z_i-z_j)}{z_i-z_j},
\end{eqnarray}
Drinfel'd introduced a class of formal power series $\Phi$
on noncommutative variables over the finite alphabet $X=\{x_0,x_1\}$.
Such a power series $\Phi$ is called an {\it associator}.

Since the system (\ref{KZ}) is completely integrable then
\begin{eqnarray}
d\Omega_n-\Omega_n\wedge\Omega_n=0.
\end{eqnarray}
This is equivalent to the following braid relations
\begin{eqnarray}
[t_{i,j},t_{i,k}+t_{j,k}]=0,&&\mbox{for distinct $i,j,k$}\\[0pt]
[t_{i,j},t_{k,l}]=0,&&\mbox{for distinct $i,j,k,l$}.
\end{eqnarray}

\begin{itemize}
\item For $n=2$, $\calT_2=\{t_{1,2}\}$, one has
\begin{eqnarray}
\Omega_2(z)=\frac{t_{1,2}}{2{\rm i}\pi}\frac{d(z_1-z_2)}{z_1-z_2}
\end{eqnarray}
and a solution of (\ref{KZ}) is given by
\begin{eqnarray}
F(z_1,z_2)=(z_1-z_2)^{\frac{t_{1,2}}{2{\rm i}\pi}}.
\end{eqnarray}

\item For $n=3$, $\calT_3=\{t_{1,2},t_{1,3},t_{2,3}\}$, there are two relations~:
\begin{eqnarray}
[t_{1,3},t_{1,2}+t_{2,3}]=0&\mbox{and}&[t_{2,3},t_{1,2}+t_{1,3}]=0.
\end{eqnarray}
One has
\begin{eqnarray}
\Omega_3(z)=\frac{1}{2{\rm i}\pi}
\biggl[t_{1,2}\frac{d(z_1-z_2)}{z_1-z_2}+t_{1,3}\frac{d(z_1-z_3)}{z_1-z_3}
+t_{2,3}\frac{d(z_2-z_3)}{z_2-z_3}\biggr]
\end{eqnarray}
and a solution of (\ref{KZ}) is given by
\begin{eqnarray}
F(z_1,z_2,z_3)
=G\biggl(\frac{z_1-z_2}{z_1-z_3}\biggr)
(z_1-z_3)^{\frac{t_{1,2}+t_{1,3}+t_{2,3}}{2{\rm i}\pi}}\cr
\end{eqnarray}
where $G$ satisfies the following Fuchs differential equation with
three regular singularities at $0,1$ and $\infty$~:
\begin{eqnarray}\label{drinfeld}
dG(z)=[x_0\omega_0(z)+x_1\omega_1(z)]G(z),
\end{eqnarray}
with
\begin{eqnarray}
\omega_0(z):=\frac{dz}{z}&\mbox{and}&\omega_1(z):=\frac{dz}{1-z},\label{drinfeldform}\\
x_0:=\frac{t_{1,2}}{2{\rm i}\pi}&\mbox{and}&x_1:=-\frac{t_{2,3}}{2{\rm i}\pi}\label{drinfeldletter}.
\end{eqnarray}
\end{itemize}
In the sequel, $X^*$ denotes the set of words defined over $X=\{x_0,x_1\}$.

\begin{proposition}[\cite{orlando}]\label{sol}
If $G(z)$ and $H(z)$ are exponential solutions of (\ref{drinfeld})
then there exists a Lie series $C\in\LCXX$ such that $G(z)=H(z)\exp(C)$.
\end{proposition}
\begin{proof}
Since $H(z)H(z)^{-1}=1$ then by differentiating, we have
$d[H(z)]H(z)^{-1}=-H(z)d[H(z)^{-1}]$. Therefore if $H(z)$ is
solution of the Drinfel'd equation then
\begin{eqnarray*}
d[H(z)^{-1}]&=&-H(z)^{-1}[dH(z)]H(z)^{-1}\\
  &=&-H(z)^{-1}[x_0\omega_0(z)+x_1\omega_1(z)],\\
  d[H(z)^{-1}G(z)]&=&H(z)^{-1}[dG(z)]+[dH(z)^{-1}]G(z)\\
  &=&H(z)^{-1}[x_0\omega_0(z)+x_1\omega_1(z)]G(z)\\
  &-&H(z)^{-1}[x_0\omega_0(z)+x_1\omega_1(z)]G(z).
\end{eqnarray*}
By simplification, we deduce that $H(z)^{-1}G(z)$ is a constant
formal power series.
Since the inverse and the product of group-like elements
is group-like then we get the expected result.
\end{proof}

\section{Iterated integral and Chen generating series}
The iterated integral associated with $w=x_{i_1}\cdots x_{i_{k}}\in X^*$,
over $\omega_0$ and $\omega_1$ and along the path $z_0\path z$,
is defined by the following multiple integral
\begin{eqnarray}
\int_{z_0}^{z}\ldots
\int_{z_0}^{z_{k-1}}\omega_{i_1}(t_1)\ldots\omega_{i_k}(t_k),
\end{eqnarray}
where $t_1\cdots t_{r-1}$ is a subdivision of the path $z_0\path z$.
In an abbreviated  notation, we denote this integral by
$\alpha_{z_0}^z(w)$ and $\alpha_{z_0}^z(1_{X^*})=1$.
\begin{example}
\begin{eqnarray*}
\alpha_0^z(x_0x_1)
&=&\int_0^z\int_0^s\omega_0(s)\omega_1(t)\\
&=&\int_0^z\int_0^s\Frac{ds}s\Frac{dt}{1-t}\\
&=&\int_0^z\Frac{ds}s\int_0^s{dt}\sum_{k\ge0}t^k\\
&=&\sum_{k\ge1}\int_0^z{ds}\frac{s^{k-1}}k\\
&=&\sum_{k\ge1}\frac{z^k}{k^2}.
\end{eqnarray*}
The last sum is nothing other than the Taylor expansion of the dilogarithm $\Li_2(z)$.
\end{example}
\begin{example}
In the same way the classical polylogarithm of order $n\ge1$
is the iterated integral associated with $x_0^{n-1}x_1$,
over $\omega_0$ and $\omega_1$ and along the path $0\path z$~:
\begin{eqnarray*}
\Li_n(z)
=\sum_{k\ge1}\frac{z^k}{k^n}
=\alpha_0^z(x_0^{n-1}x_1).
\end{eqnarray*}
Generalizing  to multi-indices $(n_1,\ldots,n_r)$, one has~:
\begin{eqnarray*}
\Li_{n_1,\ldots,n_r}(z)
=\sum_{k_1>\ldots>k_r>0}\frac{z^{k_1}}{k_1^{n_1}\ldots k_r^{n_r}}
=\alpha_0^z(x_0^{n_1-1}x_1\ldots x_0^{n_r-1}x_1).
\end{eqnarray*}
This provides an analytic prolongation of $\Li_{n_1,\ldots,n_r}$
over the Riemann surface of $\C\setminus\{0,1\}$.
\end{example}

The Chen generating series along the path $z_0\path z$
associated with $\omega_0,\omega_1$ is the following power series
\begin{eqnarray}
S_{z_0\path z}=\sum_{w\in X^*}\alpha_{z_0}^z(w)\;w
\end{eqnarray}
and it is a solution of the differential equation (\ref{drinfeld})
with the initial condition
\begin{eqnarray}
S_{z_0\path z}(z_0)=1.
\end{eqnarray}

Any Chen generating series $S_{z_0\path z}$ is  group-like \cite{ree}
 and depends only on the homotopy class of ${z_0\path z}$ \cite{chen}.
The product of two Chen generating series $S_{z_1\path z_2}$
and $S_{z_0\path z_1}$ is the Chen generating series
\begin{eqnarray}
S_{z_0\path z_2}=S_{z_1\path z_2}S_{z_0\path z_1}.
\end{eqnarray}

\section{Polylogarithm, harmonic sum and polyzetas}

To the polylogarithm $\Li_{n_1,\ldots,n_r}(z)$ we associate
the following ordinary generating series
\begin{equation}
\P_{n_1,\ldots,n_r}(z)=\frac{\Li_{n_1,\ldots,n_r}(z)}{1-z}
=\sum_{N\ge0}\H_{n_1,\ldots,n_r}(N)\;z^N,
\end{equation}
where
\begin{eqnarray}
\H_{n_1,\ldots,n_r}(N)
=\sum_{N\ge k_1>\ldots>k_r>0}\frac1{k_1^{n_1}\ldots k_r^{n_r}}
\end{eqnarray}
For $n_1>1$, the limit of $\Li_{n_1,\ldots,n_r}(z)$ and of $\H_{n_1,\ldots,n_r}(N)$,
for $z\rightarrow1$ and $N\rightarrow\infty$  exist and, by Abel's Theorem, are equal~:
\begin{eqnarray}
	\lim_{z\rightarrow1}\Li_{n_1,\ldots,n_r}(z)
	=\lim_{N\rightarrow\infty}\H_{n_1,\ldots,n_r}(N)
	=\zeta(n_1,\ldots,n_r),
\end{eqnarray}
where $\zeta(n_1,\ldots,n_r)$ is the convergent polyzeta
 \begin{eqnarray}
	\zeta(n_1,\ldots,n_r)=\sum_{k_1>\ldots>k_r>0}\frac1{k_1^{n_1}\ldots k_r^{n_r}}.
\end{eqnarray}

\begin{definition}
Let $\calZ$ be the $\Q$-algebra generated by convergent polyzetas
and let $\calZ'$ be the
$\Q[\gamma]$-algebra\footnote{Here, $\gamma$ stands for the Euler constant
$$\gamma=.57721 56649 01532 86060 65120 90082 40243 10421 59335 93992 35988 05767 23488 48677\ldots$$} generated by $\calZ$.
\end{definition}

To any multi-index ${\bf n}=(n_1,\ldots,n_r)$ corresponds the word $v=y_{n_1}\ldots y_{n_r}$
over the infinite alphabet $Y=\{y_k\}_{k\ge1}$. The word $v$ itself corresponds  to the word
ending by the letter $x_1$, $u=x_0^{n_1-1}x_1\ldots x_0^{n_r-1}x_1\in X^*x_1$. Then it is
usual to index the polylogarithms, harmonic sums and polyzetas by words (over $X$ or $Y$)~:
\begin{eqnarray}
\Li_{{\bf n}}=\Li_{v}(z)=\Li_{u},\\
\P_{{\bf n}}=\P_{v}(z)=\P_{u},\\
\H_{{\bf n}}=\H_{v}(N)=\H_{u},\\
\zeta({\bf n})=\zeta(v)=\zeta(u).
\end{eqnarray}
Let us extend, over $X^*$, the definition of $\{\Li_w\}_{w\in X^*x_1}$
and $\{\P_w\}_{w\in X^*x_1}$ by putting
\begin{eqnarray}
	\forall k\ge0,&\Li_{x_0^k}(z)=\Frac{\log^k(z)}{k},
	&\P_{x_0^k}(z)=\Frac{\Li_{x_0^k}(z)}{1-z}.
\end{eqnarray}
We get the following structures~:

\begin{theorem}[\cite{FPSAC98}]
$(\C\{\Li_w\}_{w\in X^*},.)\cong(\C\langle X\rangle,\shuffle)$.
\end{theorem}

\begin{theorem}[\cite{words03}]
$(\C\{\P_w\}_{w\in Y^*},\odot)\cong(\C\langle Y\rangle,\stuffle)$.
\end{theorem}

Extended to $\calC=\C[z,z^{-1},(1-z)^{-1}]$, we also get as a consequence
\begin{itemize}
\item The polylogarithms $\{\Li_w\}_{w\in X^*}$ (resp. $\{\P_w\}_{w\in Y^*}$)
are $\calC$-linearly independent.
Then the harmonic sums $\{\H_w\}_{w\in Y^*}$ are linearly independent.
\item The polylogarithms $\{\Li_l\}_{l\in\Lyn X}$ (resp. $\{\P_l\}_{l\in\Lyn Y}$)
are $\calC$-algebraically independent.
Then the harmonic sums $\{\H_l\}_{l\in\Lyn Y}$ are algebraically independent.
\item The polyzetas $\{\zeta(l)\}_{l\in\Lyn X\setminus\{x_0,x_1\}(\mbox{resp. }\Lyn Y\setminus\{y_1\})}$,
are generators of $\calZ$.
\end{itemize}

\begin{definition}
We put
\begin{eqnarray*}
\mathrm{L}(z):=\Sum_{w\in X^*}\Li_{w}(z)\;w&\mbox{and}&
\H(N):=\Sum_{w\in Y^*}\H_w(N)\;w.
\end{eqnarray*}
\end{definition}

The noncommuting generating series of the polylogarithms
is a solution of the differential equation (\ref{drinfeld})
with the following boundary condition
\begin{eqnarray}\label{boundarycondition}
\mathrm{L}(z)\;{}_{\widetilde{z\rightarrow0}}\;\exp(x_0\log z).
\end{eqnarray}
It follows that $\mathrm{L}$ is group-like and then $\H$ is also group-like, {\it i.e.}
\begin{eqnarray}
\Delta_{\shuffle}(\mathrm{L})=\mathrm{L}\otimes\mathrm{L}
&\mbox{and}&
\Delta_{\stuffle}(\H)=\H\otimes\H,
\end{eqnarray}
and we have
\begin{theorem}\label{grouplike}
\begin{eqnarray*}
\mathrm{L}(z)=e^{x_1\log\frac1{1-z}}\mathrm{L}_{\reg}(z)e^{x_0\log z}&\mbox{and}&
\H(N)=e^{\H_1(N)\;y_1}\H_{\reg}(N),
\end{eqnarray*}
where, denoting by $\{\hat l\}_{\Lyn X}$ (resp. $\{\hat l\}_{\Lyn Y}$)
the dual basis of ${\Lyn X}$ (resp. ${\Lyn Y}$),
\begin{eqnarray*}
\mathrm{L}_{\reg}(z):=\Prod_{l\in\Lyn X, l\neq x_0,x_1}^{\searrow}e^{\Li_l(z)\;\hat l}
&\mbox{and}&
\H_{\reg}(N):=\Prod_{l\in\Lyn Y, l\neq y_1}^{\nearrow}e^{\H_l(N)\;\hat l}.
\end{eqnarray*}
\end{theorem}

\begin{corollary}
Let $z_0\path z$ be a differentiable path on $\C-\{0,1\}$ such
that $\mathrm{L}$ admits an analytic continuation along this path.
We have $S_{z_0\path z}=\mathrm{L}(z)\mathrm{L}(z_0)^{-1}$.
\end{corollary}

\begin{proof}
By Theorem \ref{grouplike}, $\mathrm{L}(z)$ is group-like. Hence $\mathrm{L}(z_0)$ is also group-like
as is $\mathrm{L}(z_0)^{-1}$. Since the power series $\mathrm{L}(z)\mathrm{L}(z_0)^{-1}$ and $S_{z_0\path z}$
satisfy (\ref{drinfeld}) taking the same value at $z_0$ then we get the expected result.
\end{proof}

\begin{definition}
We put
\begin{eqnarray*}
Z_{\shuffle}:=\mathrm{L}_{\reg}(1)&\mbox{and}&Z_{\stuffle}:=\H_{\reg}(\infty).
\end{eqnarray*}
\end{definition}
These two noncommuting generating series $Z_{\shuffle}$ and $Z_{\stuffle}$
induce the two following regularization morphisms respectively~:

\begin{theorem}[\cite{SLC44}]\label{reg1}
Let $\zeta_{\ministuffle}:(\CYY,\ministuffle)\rightarrow(\R,.)$ be the morphism satisfying
\begin{itemize}
\item for $u,v\in Y^*,\zeta_{\ministuffle}(u\ministuffle v)
=\zeta_{\ministuffle}(u)\zeta_{\ministuffle}(v)$,
\item for all convergent words $w\in Y^*- y_1Y^*,\zeta_{\ministuffle}(w)=\zeta(w)$,
\item $\zeta_{\ministuffle}(y_1)=0$.
\end{itemize}
Then
$$\sum_{w\in X^*}\zeta_{\ministuffle}(w)\;w=Z_{\ministuffle}.$$
\end{theorem}

\begin{corollary}[\cite{SLC44}]\label{zetareg1}
For any $w\in X^*, \zeta_{\ministuffle}(w)$ belongs to the algebra $\calZ$.
\end{corollary}

\begin{theorem}[\cite{SLC44}]\label{reg2}
Let $\zeta_{\minishuffle}:(\CXX,\minishuffle)\rightarrow(\R,.)$ be the morphism verifying
\begin{itemize}
\item for $u,v\in X^*,\zeta_{\minishuffle}(u\minishuffle v)
=\zeta_{\minishuffle}(u)\zeta_{\minishuffle}(v)$,
\item for all convergent word $w\in x_0X^*x_1,\zeta_{\minishuffle}(w)=\zeta(w)$,
\item $\zeta_{\minishuffle}(x_0)=\zeta_{\minishuffle}(x_1)=0$.
\end{itemize}
Then
$$\sum_{w\in X^*}\zeta_{\minishuffle}(w)\;w=Z_{\minishuffle}.$$
\end{theorem}

\begin{corollary}[\cite{SLC44}]\label{zetareg2}
For any $w\in Y^*,\zeta_{\minishuffle}(w)$ belongs to the algebra $\calZ$.
\end{corollary}

\section{Group of associators theorem}

Drinfel'd proved that  Eqn.(\ref{drinfeld}) admits two particular solutions
on the domain $\C\setminus]-\infty,0]\cup[1,+\infty[$,
\begin{eqnarray}
G_0(z)\;{}_{\widetilde{z\path0}}\;\exp[x_0\log(z)]
\ \mbox{and}\
G_1(z)\;{}_{\widetilde{z\path1}}\;\exp[-x_1\log(1-z)].
\end{eqnarray}
He also proved that there exists $\PKZ$ such that
\begin{eqnarray}
G_0(z)=G_1(z)\PKZ.
\end{eqnarray}
L\^e and Murakami expressed the coefficients of the Drinfel'd associator
$\PKZ$ in terms of convergent polyzetas \cite{lemurakami}.

Let $\rho$ be the monoid morphism verifying
\begin{eqnarray}
\rho(x_0)=-x_1&\mbox{and}&\rho(x_1)=-x_0.
\end{eqnarray}
We also have \cite{FPSAC99}

\begin{eqnarray}\label{sigma}
\mathrm{L}(z)=\rho[\mathrm{L}(1-z)]Z_{\minishuffle}
=e^{x_0\log z}\rho[\mathrm{L}_{\reg}(1-z)]e^{-x_1\log(1-z)}Z_{\minishuffle}.
\end{eqnarray}
Thus,
\begin{eqnarray}\label{asymptoticbehaviour}
\mathrm{L}(z)\;{}_{\widetilde{z\rightarrow1}}\;\exp(-x_1\log(1-z))\;Z_{\minishuffle}.
\end{eqnarray}
It follows from Eqn.(\ref{boundarycondition}) and Eqn. (\ref{asymptoticbehaviour}),  with reference  \cite{SLC44}, that
\begin{eqnarray}
\PKZ=Z_{\minishuffle}\ ,
\end{eqnarray}
and it is group-like, {\it i.e.}
\begin{eqnarray}
\Delta_{\shuffle}(\PKZ)=\PKZ\otimes\PKZ,
\end{eqnarray}
and it can be graded in the adjoint basis \cite{orlando} as follows
\begin{eqnarray}
\PKZ=\sum_{k\ge0\atop l_1,\cdots,l_k\ge0}
\zeta_{\minishuffle}(x_1x_0^{l_1}\circ\cdots\circ x_1x_0^{l_k})
\Prod_{i=0}^k\ad_{x_0}^{l_i}x_1,
\end{eqnarray}
where, for any $l\in\N$ and $P\in\CX$, $\circ$ is defined by
\begin{eqnarray}
x_1x_0^l\circ P=x_1(x_0^l\shuffle P),
\end{eqnarray}
and $\ad_{x_0}^{l}x_1$ is the iterated Lie bracket
\begin{eqnarray}
\ad_{x_0}^0x_1=x_1&\mbox{and}&
\ad_{x_0}^{l}x_1=[x_0,\ad_{x_0}^{l-1}x_1].
\end{eqnarray}
Using the following expansion \cite{bourbaki3}
\begin{eqnarray}
\ad^n_{x_0}x_1=\sum_{i=0}^n{i\choose n}x_0^{n-i}x_1x_0^i,
\end{eqnarray}
one then gets, via the regularization process of Theorem \ref{reg2},
the expression for the Drinfel'd associator $\Phi_{KZ}$
given by L\^e and Murakami \cite{lemurakami}.

Finally, the asymptotic behaviour of $\mathrm{L}$ on (\ref{asymptoticbehaviour}) leads to
\begin{proposition}[\cite{FPSAC98}]\label{chenregularization}
\begin{eqnarray*}
S_{\eps\path1-\eps}\;{}_{\widetilde{\eps\rightarrow0^+}}\;
e^{-x_1\log\eps}\;Z_{\minishuffle}\;e^{-x_0\log\eps}.
\end{eqnarray*}
\end{proposition}
In other words, $Z_{\minishuffle}$ is the regularized Chen generating series
$S_{\eps\path1-\eps}$ of differential forms $\omega_0$ and $\omega_1$~:
$Z_{\minishuffle}$ is the noncommutative generating series of the finite parts of the
Chen generating series $e^{x_1\log\eps}\;S_{\eps\path1-\eps}\;e^{x_0\log\eps}$,
the concatenation of $e^{x_0\log\eps}$ and then $S_{\eps\path1-\eps}$ and finally,
$S_{\eps\path1-\eps}$.

Let $\{b_{n,k}(t_1,\ldots,t_{n-k+1})\}_{n\ge k\ge1\atop n\in\N_+}$
be the Bell polynomials given by the following exponential generating series\footnote{Compare this expression with the less general Eq.(\ref{bellpol1}). Here $b(n,k)$ is a refinement by the number of blocks ({\em i.e.} $k$) of Eq.(\ref{bellpol1}).  Thus one has
$$
B_n(X_1,X_2,\cdots X_n)=\sum_{k=0}^n b_{n,k}(X_1,\ldots,X_{n-k+1})\ .
$$}
\begin{eqnarray}
\sum_{n,k=0}^{\infty}b_{n,k}(t_1,\ldots,t_{n-k+1})\frac{v^nu^k}{n!}
=\exp\biggl(u\sum_{l=0}^{\infty}t_l\frac{v^l}{l!}\biggr).
\end{eqnarray}
We  specify the variables $\{t_l\}_{l\ge0}$
\begin{eqnarray}
t_1=\gamma,&\mbox{for $l>1$,}&t_l=(-1)^{l-1}(l-1)!\zeta(l),
\end{eqnarray}
then  let $B(x_1)$ be the following power series
\begin{eqnarray}
B(x_1)=1+\Sum_{n\ge1}\Sum_{k=1}^nb_{n,k}(\gamma,-\zeta(2),2\zeta(3),\ldots)
\frac{(-x_1)^n}{n!}
\end{eqnarray}
and
\begin{eqnarray}
B'(x_1)=e^{-\gamma\;x_1}B(x_1).
\end{eqnarray}

We get

\begin{proposition}
Let $\Psi_{KZ}$ be an element of the quasi-shuffle algebra,
$(\C\langle\!\langle Y\rangle\!\rangle,\stuffle)$, such that
$$\Pi_X\Psi_{KZ}=B(x_1)\PKZ,$$
where $\Pi_X\Psi_{KZ}$ is the projection of $\Psi_{KZ}$ over $X$.
Then $\Psi_{KZ}$ is group-like, {\it i.e.}
\begin{eqnarray}
\Delta_{\stuffle}(\Psi_{KZ})=\Psi_{KZ}\otimes\Psi_{KZ},
\end{eqnarray}
and satisfies
\begin{itemize}
\item $\langle\Psi_{KZ}\bv1_{Y^*}\rangle=1$,
\item $\langle\Psi_{KZ}\bv y_1\rangle=\gamma$,
\item for any $r_1>1$,
$\langle\Psi_{KZ}\bv y_{r_1}\ldots y_{r_k}\rangle=\zeta(r_1,\ldots,r_k),$
\item for any $u,v\in Y^*,
\langle\Psi_{KZ}|u\stuffle v\rangle
=\langle\Psi_{KZ}| u\rangle\langle\Psi_{KZ}|v\rangle.$
\end{itemize}
\end{proposition}

Therefore we obtain,
\begin{proposition}
The noncommutative generating series $\Psi_{KZ}$ can be
factorized by Lyndon words as follows
\begin{eqnarray*}
\Psi_{KZ}=e^{\gamma\;y_1}Z_{\ministuffle}.
\end{eqnarray*}
\end{proposition}

More generally, we have

\begin{theorem}\label{associatorgp}
Let $A$ be a commutative $\Q$-algebra.
For any $\Phi\in\AXX$ and $\Psi\in A\langle\!\langle Y\rangle\!\rangle$
such that $\Pi_X\Psi=B(x_1)\Phi$,
there exists a unique $C\in\LAXX$
with coefficients in $A$ such that
\begin{eqnarray*}
\Phi=\Phi_{KZ}e^C&\mbox{and}&\Psi=B(y_1)\Pi_Y(\Phi_{KZ}e^C),
\end{eqnarray*}
where $\Pi_Y(\Phi_{KZ}e^C)$ is the projection of $\Phi_{KZ}e^C$ over $Y$.
\end{theorem}

\begin{proof}
Let $C\in\LAXX$. Then, by Proposition \ref{sol}, $\mathrm{L}'=\mathrm{L}e^C$ is solution of (\ref{drinfeld}).
Let $\H'$ be the noncommutative generating series of the Taylor coefficients,
belonging to the harmonic algebra, of $\{(1-z)^{-1}\langle\mathrm{L}'\bv w\rangle\}_{w\in Y^*}$.
Then $\H'(N)$ is also group-like.

By the asymptotic expansion on (\ref{asymptoticbehaviour}), we get
\begin{eqnarray*}
{\mathrm{L}'(z)}&{}_{\widetilde{\eps\to1}}&e^{-x_1\log(1-z)}Z_{\minishuffle}e^C.
\end{eqnarray*}
We then put $\Phi:=Z_{\minishuffle}e^C$ and deduce that
\begin{eqnarray*}
\frac{\mathrm{L}'(z)}{1-z}\;{}_{\widetilde{z\to1}}\;\mathrm{Mono}(z)\Phi
&\mbox{and}&
\H'(N)\;{}_{\widetilde{N\to\infty}}\;\mathrm{Const}(N)\Pi_Y\Phi,
\end{eqnarray*}
where \cite{cade}
\begin{eqnarray*}
\mathrm{Mono}(z)=e^{-(x_1+1)\log(1-z)}=\Sum_{k\ge0}\P_{y_1^k}(z)\;y_1^k\\ \mathrm{Const}=\sum_{k\ge0}\H_{y_1^k}\;y_1^k=\exp\biggl[-\Sum_{k\ge1}\H_{y_k}\Frac{(-y_1)^k}{k}\biggr].
\end{eqnarray*}
Let $\kappa_w$ be the constant part of $\H'_w(N)$. Then,
\begin{eqnarray*}
\sum_{w\in Y^*}\kappa_w\;w=B(y_1)\Pi_Y\Phi.
\end{eqnarray*}
We now put $\Psi:=B(y_1)\Pi_Y(Z_{\minishuffle}e^C)$
(and also $\Psi':=B'(y_1)\Pi_Y(Z_{\minishuffle}e^C)$.
\end{proof}

Thus, any associator $\Phi$ and its image $\Psi$,
can be determined from the Drinfel'd associator $\PKZ$
and its image $\Psi_{KZ}$ respectively, by the action of the group
of constant Lie exponential series over $X$. This group is nothing
other than  the differential Galois group of the differential equation
(\ref{drinfeld}) and contains  in particular \cite{orlando}
the monodromy group of (\ref{drinfeld}) given by \cite{FPSAC98}
\begin{eqnarray}\label{m_0,m_1}
\calM_0=e^{2i\pi x_0}&\mbox{and}&
\calM_1=\Phi_{KZ}^{-1}e^{-2i\pi x_1}\Phi_{KZ}.
\end{eqnarray}

By Proposition \ref{chenregularization}, we already saw that $Z_{\minishuffle}$
is the concatenation of the Chen generating series \cite{chen}
$e^{x_0\log\varepsilon}$ and then $S_{\eps\path1-\eps}$ and finally, $e^{x_1\log\varepsilon}$~:
\begin{eqnarray}
Z_{\minishuffle}\;{}_{\widetilde{\eps\rightarrow0^+}}\;
e^{x_1\log\eps}S_{\eps\path1-\eps}\;e^{x_0\log\eps}.
\end{eqnarray}
From (\ref{m_0,m_1}), the action of the monodromy group gives
\begin{eqnarray}
e^{x_1\;2k_1\mathrm{i}\pi}Z_{\minishuffle}e^{x_0\;2k_0\mathrm{i}\pi}
\;{}_{\widetilde{\eps\rightarrow0^+}}\;
e^{x_1(\log\eps+2k_1\mathrm{i}\pi)}
S_{\eps\path1-\eps}e^{x_0(\log\eps+2k_0\mathrm{i}\pi)}
\end{eqnarray}
as being the concatenation of the Chen generating series \cite{chen}
$e^{x_0(\log\varepsilon+2k_0\mathrm{i}\pi)}$
(along a circular path turning $k_0$ times around $0$),
then the Chen generating series $S_{\eps\path1-\eps}$ and finally,
the Chen generating series $e^{x_1(\log\varepsilon+2k_1\mathrm{i}\pi)}$
(along a circular path turning $k_1$ times around $1$).
More generally, the action of the Galois differential group
of polylogarithms shows  that, for any Lie series $C$,
the associator $\Phi=Z_{\minishuffle}e^C$
is the concatenation of some Chen generating series $e^C$
and $e^{x_0\log\varepsilon}$; then the Chen generating series $S_{\eps\path1-\eps}$
and finally, $e^{x_1\log\varepsilon}$~:
\begin{eqnarray}\label{concatenationofchenseries}
\Phi\;{}_{\widetilde{\eps\rightarrow0^+}}\;
e^{x_1\log\eps}S_{\eps\path1-\eps}e^{x_0\log\eps}\;e^C.
\end{eqnarray}
As in the proof of Theorem \ref{associatorgp}, by construction the associator $\Phi$
is then the noncommutative generating series of the finite parts of the coefficients
of the Chen generating series $S_{z_0\path1-z_0}e^C$, for $z_0=\eps\rightarrow0^+$.

\begin{lemma}
Let $\Phi\in dm(A)=\{\PKZ e^C\bv C\in\LAXX\mbox{ and }
\langle e^C\bv\epsilon\rangle=1,\langle e^C\bv x_0\rangle=\langle e^C\bv x_1\rangle=0\}$.
Let $\Pi_Y\Phi$ be the projection of $\Phi$ over $Y$. One has
\begin{eqnarray*}
\Psi=B(y_1)\Pi_Y\Phi&\iff&\Psi'=B'(y_1)\Pi_Y\Phi.
\end{eqnarray*}
In particular,
\begin{eqnarray*}
\Psi_{KZ}=B(y_1)\Pi_Y\Phi_{KZ}&\iff&
\Psi'_{KZ}=B'(y_1)\Pi_Y\Phi_{KZ}.
\end{eqnarray*}
\end{lemma}
{\section{The route to properties of Euler's $\gamma$ constant}
We have already mentioned that as a remarkable by-product of our extension of the algebraic considerations related to simple quantum mechanics, we are led to algebras involving polyzeta functions.  This  excursion is  not surprising, as it is already well known that evaluations in quantum field theory involve these functions \cite{kre}. In this context, an important related number    is the Euler constant  $\gamma$. We are thus led to  to consider properties of $\gamma$.  In fact, in this note we  arrive at a proposition concerning the rationality of $\gamma$.  The basis of this result depends on the following conjectures\footnote{For a discussion and proof of these conjectures see \cite{conj}.}
\begin{conjecture}\label{lem2}
Let $\Phi\in dm(A)$ and $\Psi'=B'(y_1)\Pi_Y\Phi$.
The local coordinates (of the second kind) of $\Phi$ (resp. $\Psi'$), in the Lyndon-PBW basis,
are polynomials of generators $\{\zeta(l)\}_{l\in\Lyn X}-\{x_0,x_1\}$
(resp. $\{\zeta(l)\}_{l\in\Lyn Y}-\{y_1\}$) of $\calZ$.
While $C$ describes $\LAXX$, these local coordinates
describe $A[\{\zeta(l)\}_{l\in\Lyn X}-\{x_0,x_1\}]$
(resp. $A[\{\zeta(l)\}_{l\in\Lyn Y}-\{y_1\}]$).
\end{conjecture}

\begin{conjecture}
For any $\Phi\in dm(A)$, by identifying the local coordinates (of the second kind) on two
members of the identities $\Psi=B(y_1)\Pi_Y\Phi$, or equivalently on
$\Psi'=B'(y_1)\Pi_Y\Phi$, we get polynomial relations with
coefficients in $A$ of the convergent polyzetas.
\end{conjecture}

Therefore, if the two preceding conjectures prove to be true, we get the route to a striking result through the complete description of the ideal of relations between polyzetas

\begin{state}\label{alginpdt}
While $\Phi$ describes $dm(A)$, the identities
$\Psi=B(y_1)\Pi_Y\Phi$ describe the ideal of polynomial
relations of coefficients in $A$ of the  convergent polyzetas.

Moreover, if the Euler constant $\gamma$ does not belong to $A$
then these relations are algebraically independent of $\gamma$.
\end{state}

\begin{state}\label{notalgebraic}
If $\gamma\notin A$ then $\gamma\notin\bar A$.
\end{state}

With $A=\Q$, it follows immediately that

\begin{state}\label{petitpoisson}
$\gamma$ is not an algebraic irrational number.
\end{state}

\begin{state}\label{grospoisson}
$\gamma$ is a rational number.
\end{state}

\begin{proof}
Since $\gamma$ satisfies $t^2-\gamma^2=0$ then $\gamma$
is algebraic over $A=\Q(\gamma^2)$.
If $\gamma$ is transcendental then $\gamma\notin A=\Q(\gamma^2)$.
Using Corollary \ref{notalgebraic}, with $A=\Q(\gamma^2)$,
$\gamma$ is not algebraic over $A=\Q(\gamma^2)$.
This contradicts the previous assertion.
Thus, by Corollary \ref{petitpoisson}, it is established that
$\gamma$ is rational over $\Q$.
\end{proof}

\section{Structure of polyzetas}
Let $\Phi\in dm(A)$ and let $\Psi=B'(y_1)\Pi_Y\Phi$. We introduce two algebra morphisms
\begin{eqnarray}
\begin{array}{ccc}
\phi:(A\langle X\rangle,\shuffle)&\longrightarrow&A,\\
u&\longmapsto&\langle\Phi\bv u\rangle,
\end{array}
\end{eqnarray}
and
\begin{eqnarray}
\begin{array}{ccc}
\psi:(A\langle Y\rangle,\stuffle)&\longrightarrow&A,\\
v&\longmapsto&\langle\Psi\bv v\rangle,
\end{array}
\end{eqnarray}
satisfying respectively
\begin{eqnarray}
\phi(\epsilon)=1&\mbox{and}&\phi(x_0)=\phi(x_1)=0,\\
\psi(\epsilon)=1&\mbox{and}&\psi(y_1)=0.
\end{eqnarray}
Hence,
\begin{eqnarray}\label{phipsi}
\Phi=\sum_{u\in X^*}\phi(u)\;u
=\prod_{l\in\Lyn X\atop l\neq x_0,x_1}^{\searrow}e^{\phi(l)\;\hat l},&&\\
\Psi=\sum_{v\in Y^*}\psi(u)\;u
=\prod_{l\in\Lyn Y\atop l\neq y_1}^{\nearrow}e^{\psi(l)\;\hat l}.&&
\end{eqnarray}

Using these factorizations of the monoids by Lyndon words, we get

\begin{state}
For any $\Phi\in dm(A)$, let $\Psi=B'(y_1)\Pi_Y\Phi$. Then
\begin{eqnarray*}
\Prod_{l\in\Lyn X\atop l\neq x_0,x_1}^{\searrow}e^{\phi(l)\;\hat l}
=e^{-\Sum_{k\ge2}\zeta(k)\Frac{(-x_1)^k}{k}}
\Pi_X\Prod_{l\in\Lyn Y\atop l\neq y_1}^{\nearrow}e^{\psi(l)\;\hat l}.
\end{eqnarray*}
In particular, if $\Phi=Z_{\minishuffle}$ and $\Psi=Z_{\ministuffle}$ then
\begin{eqnarray*}
\Prod_{l\in\Lyn X\atop l\neq x_0,x_1}^{\searrow}e^{\zeta(l)\;\hat l}
=e^{-\Sum_{k\ge2}\zeta(k)\Frac{(-x_1)^k}{k}}
\Pi_X\Prod_{l\in\Lyn Y\atop l\neq y_1}^{\nearrow}e^{\zeta(l)\;\hat l}.
\end{eqnarray*}
\end{state}

Since
\begin{eqnarray}
\forall l\in\Lyn Y\iff\Pi_Xl\in\Lyn X\setminus\{x_0\}
\end{eqnarray}
then identifying the local coordinates,
we get polynomial relations among the generators
which are algebraically independent of $\gamma$.

\begin{state}\label{ideal}
For  $\ell\in\Lyn Y-\{y_1\}$ (resp. $\Lyn X-\{x_0,x_1\}$),
let $P_\ell\in\LQX$ (resp. $\LQY$) be the decomposition
of the polynomial $\Pi_X\hat\ell\in\QX$ (resp. $\Pi_Y\hat\ell\in\QY$)
in the Lyndon-PBW basis $\{\hat l\}_{l\in\Lyn X}$ (resp. $\{\hat l\}_{l\in\Lyn Y}$)
and let $\check P_\ell\in\Q[\Lyn X-\{x_0,x_1\}]$ (resp. $\Q[\Lyn Y-\{y_1\}]$) be its dual.
Then one obtains
\begin{eqnarray*}
\Pi_X\ell-\check P_\ell\in\ker\phi&(\mbox{resp.}
&\Pi_Y\ell-\check P_\ell\in\ker\psi).
\end{eqnarray*}
In particular, for $\phi=\zeta$ (resp. $\psi=\zeta$) then one also obtains
\begin{eqnarray*}
\Pi_X\ell-\check P_\ell\in\ker\zeta&(\mbox{resp.}
&\Pi_Y\ell-\check P_\ell\in\ker\zeta).
\end{eqnarray*}
Moreover, for any $\ell\in\Lyn Y-\{y_1\}$ (resp. $\Lyn X-\{x_0,x_1\}$),
the polynomial $\Pi_Y\ell-\check P_\ell\in\QY$ (resp. $\QX$)
is homogenous of degree equal to $\bv\ell\bv>1$.
\end{state}

\begin{state}[Structure of polyzetas]
The $\Q$-algebra generated by convergent polyzetas
is isomorphic to the {\em graded} algebra $(\Q\oplus(Y-y_1)\QY/\ker\zeta,\stuffle)$.
\end{state}

\begin{proof}
Since $\ker\zeta$ is an ideal generated by the homogenous polynomials
then the quotient $\Q\oplus(Y-y_1)\QY/\ker\zeta$ is graded.
\end{proof}

\begin{state}
The $\Q$-algebra of polyzetas is freely generated by {\em irreducible} polyzetas.
\end{state}
\begin{proof}
For any $\lambda\in\Lyn Y$, if $\lambda=\check P_\lambda$ then one obtains the conclusion otherwise
$\Pi_X\lambda-\check P_\lambda\in\ker\zeta$.
Since $\check P_\lambda\in\Q[\Lyn X]$ then $\check P_\lambda$ is polynomial
in Lyndon words of degree less than or equal to $|\lambda|$. Since each Lyndon word
does appear in this decomposition of $\check P_\lambda$, after applying $\Pi_Y$,
we may repeat the same process  until we obtain irreducible polyzetas.
\end{proof}

\end{document}

%% file: unlabelled11.tex
\def\JPicScale{0.7}
\unitlength \JPicScale mm

\begin{picture}(122.5,82.5)(-5,10)
\linethickness{0.75mm}
\multiput(69.34,49.61)(0.12,0.37){76}{\line(0,1){0.37}}
\linethickness{0.75mm}
\multiput(72.23,52.11)(0.12,0.37){67}{\line(0,1){0.37}}
\linethickness{0.75mm}
\multiput(82.76,78.82)(0.12,-0.57){48}{\line(0,-1){0.57}}

\put(60,80){\circle{5}}
\put(80.13,79.87){\circle{5}}
\put(100.13,79.87){\circle{5}}
\put(120,80){\circle{5}}

\put(70.2,49.88){\circle*{5}}
\put(90.3,49.56){\circle*{5}}
\put(110,50){\circle*{5}}

\put(56,86){$\{1\}$}
\put(70,86){$\{2,3,4\}$}
\put(88,86){$\{5,6,7,8,9\}$}
\put(115,86){$\{10,11\}$}

\put(60.2,41.88){$\{2,3,5\}$}
\put(77.3,41.88){$\{1,4,6,7,8\}$}
\put(104,41.88){$\{9,10,11\}$}

\linethickness{0.75mm}
\multiput(71.97,49.74)(0.12,0.14){212}{\line(0,1){0.14}}

\linethickness{0.75mm}
\multiput(102.76,79.47)(0.12,-0.48){57}{\line(0,-1){0.48}}

\linethickness{0.75mm}
\multiput(92.76,49.61)(0.12,0.36){81}{\line(0,1){0.36}}
\linethickness{0.75mm}
\multiput(90.26,48.29)(0.12,0.37){78}{\line(0,1){0.37}}
\linethickness{0.75mm}
\multiput(88.29,49.08)(0.12,0.36){81}{\line(0,1){0.36}}
\linethickness{0.75mm}
\multiput(111.97,50.66)(0.12,0.37){72}{\line(0,1){0.37}}
\linethickness{0.75mm}
\multiput(109.74,51.05)(0.12,0.37){73}{\line(0,1){0.37}}
\linethickness{0.75mm}
\multiput(61.71,78.03)(0.12,-0.13){216}{\line(0,-1){0.13}}
\end{picture}

\vspace{-2cm}

{\small{\bf Fig 1}\pointir \it Diagram from $P_1,\ P_2$ (set partitions of $[1\cdots 11]$).\\
$P_1=\left\{\{2,3,5\},\{1,4,6,7,8\},\{9,10,11\}\right\}$ and  $P_2=\left\{\{1\},\{2,3,4\},\{5,6,7,8,9\},\{10,11\}\right\}$ (respectively black spots for $P_1$ and white spots for $P_2$).\\
The incidence matrix corresponding to the diagram (as drawn) or these partitions is
${\pmatrix{0 & 2 & 1 & 0\cr 1 & 1 & 3 & 0\cr 0 & 0 & 1 & 2}}$
 but, due to the fact that the defining partitions are unordered, one can permute the spots (black and white) among themselves and so the lines and columns of this matrix can be permuted; the diagram could be represented by the matrix ${\pmatrix{0 & 0 & 1 & 2\cr 0 &  2 & 1 & 0\cr 1 & 0 & 3 & 1}}$ as well.}
